\documentclass[conference]{IEEEtran}
\usepackage{graphicx}
\usepackage{subfigure}
\usepackage[thmmarks]{ntheorem}
\usepackage{amsfonts}
\usepackage{cite}
\usepackage{amsmath, amssymb} 
\long\def\symbolfootnote[#1]#2{\begingroup
\def\thefootnote{\fnsymbol{footnote}}\footnote[#1]{#2}\endgroup}

\newtheorem{theorem}{Theorem}[section]
\newtheorem{lemma}[theorem]{Lemma}

\theoremheaderfont{\normalfont}\theorembodyfont{\normalfont}
\theoremstyle{nonumberplain}

\begin{document}
\bibliographystyle{IEEEtran}

\title{Rate Region Frontiers for $n-$user Interference Channel with Interference as Noise}

\author{
\authorblockN{Mohamad Charafeddine, Aydin Sezgin, Arogyaswami
Paulraj}
\authorblockA{Information Systems Laboratory, Department of Electrical Engineering, Stanford University\\
Packard 225, 350 Serra Mall, Stanford, CA 94305, Email:
\{mohamad,sezgin,apaulraj\}@stanford.edu }}


\maketitle

\begin{abstract}
This paper presents the achievable rate region frontiers for the
$n-$user interference channel when there is no cooperation at the
transmit nor at the receive side. The receiver is assumed to treat
the interference as additive thermal noise and does not employ
multiuser detection. In this case, the rate region frontier for the
$n-$user interference channel is found to be the union of $n$
hyper-surface frontiers of dimension $n-1$, where each is
characterized by having one of the transmitters transmitting at full
power. The paper also finds the conditions determining the convexity
or concavity of the frontiers for the case of two-user interference
channel, and discusses when a time sharing approach should be
employed with specific results pertaining to the two-user symmetric
channel. \symbolfootnote[0]{This work was supported by the NSF
DMS-0354674-001 and ONR N00014-02-0088 grants.}
\end{abstract}


\IEEEpeerreviewmaketitle


\section{Introduction}
The capacity region of a two-user communication channel has been an
open problem for about $30$ years \cite{Sato:2usersCh,
Sato:degardedGaussian2users}. Information-theoretic bounds through
achievable rate regions have been proposed, most famously with the
Han-Kobayashi region \cite{Kobayashi:interf}. The capacity of the
Gaussian interference channel under strong interference has been
found in \cite{Sato:capacityWithStrongInterf}. Recent results on the
two-user interference channel to within one bit of capacity have been shown in
\cite{EtkinTse:GaussianInterfChannelToOneBit}. The aforementioned
referenced literature focused on the two-user interference channel
from an information-theoretic point of view. This work presents the
frontiers for the achievable rate regions for $n-$user interference
channel when the interference is treated as additive noise and no
multiuser detection is employed. Examples where we encounter the
need to define such rate regions are found in multicell
communications, in addition to mesh and sensor networks where the
preference is in employing low-complexity transceivers.

The system setup is presented in section \ref{sec_SysSetup}. Section
\ref{sec_FrontierTwoUser} discusses the achievable rate frontiers
for the two-user interference channel. The $n-$user generalization
is treated in section \ref{sec_FrontierGeneral}. Section
\ref{sec_2usersRegion} focuses on characterizing the two-user
interference channel in terms of convexity or concavity and when a
time sharing approach should be used, with specific results to the
symmetric channel.
\begin{figure} [ttt]
\centering
\includegraphics[width=166pt,height=120pt]{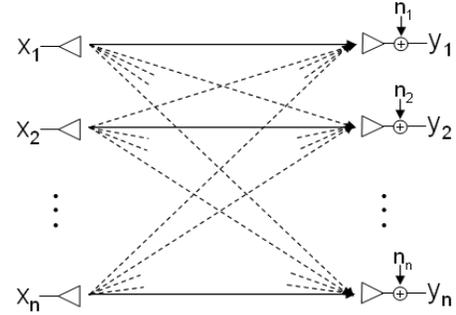}
\caption{$n-$user interference channel}
 \label{fig_n-user}
 \vspace{-.33cm}
\end{figure}

\section{System Setup}
\label{sec_SysSetup} The $n-$user interference channel is presented
in Fig. \ref{fig_n-user} with $n$ transmitters and $n$ receivers.
The $i^{th}$ transmitter transmits its signal $x_i$ to the intended
$i^{th}$ receiver with a power $P_i$. The receivers have independent
additive complex white Gaussian noise with a power noise variance of
$\sigma_n^2$. Each transmitter is assumed to have a maximum power
constraint of $P_{\max}$. No cooperation is assumed between the
nodes at the transmit side nor at the receive side. The transmitters
have a single antenna each, and they communicate over frequency flat
channels. $g_{i,j}$ denotes the channel power gain received at the
$i^{th}$ receiver from the $j^{th}$ transmitter, and there are no
constraints over the values or distributions of $g_{i,j}$. Therefore
$g_{i,i}$ is the channel gain of the $i^{th}$ desired signal, where
as $g_{i,j},~ j \neq i$ represent the interfering channel gains.
{\bf P} is the transmit power vector of length $n$, where the
$i^{th}$ element $P_i$ denotes the transmit power of the $i^{th}$
transmitter. Treating the interference as additive noise throughout
this paper and with no multiuser detection employed, $C_i$ denotes
the maximum reliable rate of communication between the $i^{th}$
transmitter and the $i^{th}$ receiver. Therefore the achievable rate
for the $i^{th}$ transmit-receive pair is written as:
\begin{align}
C_i({\bf P}) =
\log_2\left(1+\frac{g_{i,i}P_i}{\sigma_n^2+\sum_{j\neq
i}g_{i,j}P_j}\right). \label{Ci}
\end{align}
The objective of this work is to find the achievable rate region for
the $n$ transmit-receive pairs. Section \ref{sec_FrontierTwoUser}
finds the achievable rate region frontier for the two-user channel,
and section \ref{sec_FrontierGeneral} generalizes the frontier for
the $n-$user case.

\section{Achievable Rate Region Frontier for
Two-user Interference Channel} \label{sec_FrontierTwoUser} This
section studies the two-user interference channel. In this case,
(\ref{Ci}) can be expressed in function of $P_1$ and $P_2$ as
$C_i(P_1,P_2),~i=1,2$. For notational brevity, the channel power
gains are normalized by the noise variance, specifically:
\begin{align}
\begin{array}{ccc}
a=g_{1,1}/\sigma_n^2, & & c=g_{2,2}/\sigma_n^2,\\
b=g_{1,2}/\sigma_n^2, & & d=g_{2,1}/\sigma_n^2.
\end{array}\nonumber
\end{align}
\begin{figure} [ttt]
\centering
\includegraphics[width=166pt,height=57pt]{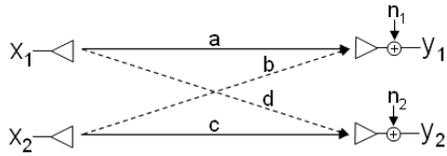}
\caption{Two-user interference channel}
 \label{fig_2-user}
  \vspace{-.33cm}
\end{figure}
The two-user interference channel is depicted in Fig.
\ref{fig_2-user}. $C_1$ and $C_2$ can therefore be written as:
\begin{align}
C_1(P_1,P_2) = \log_2\left(1+\frac{aP_1}{1+bP_2}\right),
\label{C1P1P2orig}
\end{align}
\begin{align}
C_2(P_1,P_2) = \log_2\left(1+\frac{cP_2}{1+dP_1}\right).
\label{C2P1P2orig}
\end{align}
Our objective is to find a frontier of the achievable rate region of
(\ref{C1P1P2orig}) and (\ref{C2P1P2orig}) through the power control
of $P_1$ and $P_2$, where each transmitter is subject to a maximum power
constraint of $P_{\max}$.

Fig. \ref{fig_twoUserRegion} illustrates an example of rate region
for a two-user interference channel, with $C_1$ and $C_2$ as the
x-axis and the y-axis, respectively. From (\ref{C1P1P2orig}), $C_1$ in
monotonically increasing in $P_1$ and monotonically decreasing in
$P_2$, thus the point $C_1(P_{\max},0)$, alternatively annotated as
point C on the x-axis in the figure, represents the maximum value
$C_1$ can obtain. Similarly for the y-axis, the maximum value that
$C_2$ achieves is $C_2(0,P_{\max})$, annotated as the point A on the
y-axis. The point B has the following coordinates of
$C_1(P_{\max},P_{\max})$ and $C_2(P_{\max},P_{\max})$.

\subsection{Rate Region Frontier Formulation}
\label{FrontierFormulation} The rate region frontier can be found by
setting $C_1$ to a value $R$, then $R$ is swept over the full range
of $C_1$, i.e. from $0$ to $C_1(P_{\max},0)$, while finding the
maximum $C_2$ value that can be achieved for each $R$. Hence, for a
constant rate $C_1=R$,
\begin{equation}
C_1(P_1,P_2)=R=\log_2\left(1+\frac{aP_1}{1+bP_2}\right).
\label{C1P1P2}
\end{equation}
Therefore the relation between $P_1$ and $P_2$ is obtained as
follows:
\begin{align}
P_1=\frac{1}{a}(1+bP_2)(2^R-1).
\label{P1P2relation}
\end{align}
From (\ref{P1P2relation}), $C_2(P_1,P_2)$ can now be written in
function of one parameter as $C_2(P_2)$, specifically:
\begin{align}
C_2(P_2) =
\log_2\left(1+\frac{cP_2}{1+\frac{d}{a}(1+bP_2)(2^R-1)}\right).
\label{C2P2}
\end{align}
It is important to analyze the behavior of $C_2(P_2)$ in terms of
$P_2$. This is presented in the following lemma: 
\begin{lemma}
Setting $C_1$ at a constant rate, $C_1(P_1,P_2)=R$, $C_2(P_2)$ is a
monotonically increasing function in $P_2$. \label{lemma_mono}
\end{lemma}
\begin{proof}
The proof is provided in the Appendix.
\end{proof}
\emph{Remark:} A direct implication of the monotonicity of the
relation in (\ref{C2P2}) is that if $C_2$ is equal to a constant
$C_2^*$ for the rate of $C_1=R$, then: i) there is a unique $P_2^*$
that achieves $C_2^*$, ii) when $P_2^*$ is determined, then
$P_1=P_1^*$ is uniquely defined from (\ref{P1P2relation}), iii) from
$P_1^*$ and $P_2^*$, $C_1$ is uniquely defined as $C_1=R=C_1^*$ from
(\ref{C1P1P2}). Thus $P_1^*$ and $P_2^*$ uniquely define a point in
the rate region with coordinates $C_1^*$ and $C_2^*$.

In other words, any point in the rate region is achieved solely by a
unique power tuple. This leads to what we denote by potential lines
$\Phi$ in the rate region which are formed by holding one of the
power parameters constant to a certain value, and sweeping the other
parameter over its full power range. For instance, the potential
line $\Phi(:,P_{\max})$ is formed by sweeping $P_1$ from $0$ to
$P_{\max}$ and holding $P_2$ at $P_{\max}$. Based on the uniqueness
property just discussed, these potentials lines along the $P_1$
dimension are therefore non-touching, i.e. $\Phi(:,p_2)$ and
$\Phi(:,p_2')$ do not intersect if $p_2 \neq p_2'$.

\begin{figure} [ttt]
\centering
\includegraphics[width=250pt,height=180pt]{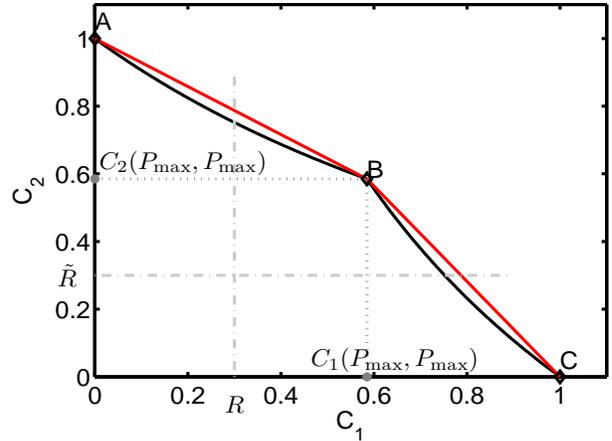}
\caption{ Two-user interference channel achievable rate region. (For
$P_{\max}=1$, $a=1$, $b=1$, $c=1$, and $d=1$)}
\label{fig_twoUserRegion} \vspace*{-0.5cm}
\end{figure}

The rate region frontier then simplifies to finding the maximum
value $C_2(P_2)$ could achieve for any value of $C_1(P_1,P_2)=R$.
Effectively, this is formulated as:
\begin{align}
\begin{array}{ll}
\arg \mathop {\max}\limits_{P_2} & C_2(P_2) \label{argC2}\\
\mbox{subject to} & C_1(P_1,P_2) = R\\
& P_i \leq P_{\max}~~~~~i=1,2.
\end{array}
\end{align}

\subsection{$C_2$ frontier for $0\leq R \leq C_1(P_{\max},P_{\max})$}
\label{subsec_Ra} As (\ref{C1P1P2}) is monotonically increasing in
$P_1$ and monotonically decreasing in $P_2$, $R$ can only exceed the
value $C_1(P_{\max},P_{\max})$ when $P_2$ is less than $P_{\max}$.
Thus $P_2=P_{\max}$ is attainable only when $0\leq R\leq
C_1(P_{\max},P_{\max})$, and $P_2$ needs to be less than $P_{\max}$
otherwise. From the proof provided in Lemma \ref{lemma_mono}, where
(\ref{C2P2}) is proved to be monotonically increasing in $P_2$, and
for the following range of $R$: $0\leq R \leq
C_1(P_{\max},P_{\max})$, the solution to (\ref{argC2}) is:
\begin{align}
\arg \max_{P_2} C_2(P_2)=P_{\max} \label{P2Pmax}
\end{align}
Therefore in this range of $R$ or equivalently of $C_1$, using
(\ref{P1P2relation}) and (\ref{P2Pmax}), $C_2$ is expressed in
function of $C_1$ as follows:
\begin{align}
C_2(C_1)=
\log_2\left(1+\frac{cP_{\max}}{1+\frac{d}{a}(1+bP_{\max})(2^{C_1}-1)}\right).
\nonumber
\end{align}

\subsection{$C_2$ frontier for $C_1(P_{\max},P_{\max}) \leq R \leq C_1(P_{\max},0)$}
\label{subsec_Rb} Using symmetry of the previous result, for a
constant rate $C_2=\tilde{R}$, there is a linear relation between
$P_1$ and $P_2$. And thus $C_1(P_1,P_2)$ can be written in function
of one parameter $P_1$ as follows:
\begin{align}
C_1(P_1) =
\log_2\left(1+\frac{aP_1}{1+\frac{b}{c}(1+dP_1)(2^{\tilde{R}}-1)}\right).
\nonumber
\end{align}
And by symmetry of the result in Lemma \ref{lemma_mono}, $C_1(P_1)$
is monotonically increasing in $P_1$. Thus by symmetry, for the
following range of $\tilde{R}$:
\begin{align}
0\leq \tilde{R} \leq C_2(P_{\max},P_{\max}), \label{RtildeRange}
\end{align}
we have:
\begin{align}
\arg \max_{P_1} C_1(P_1)=P_{\max}.
\nonumber
\end{align}
Therefore for this range of $\tilde{R}$, $P_1=P_{\max}$ is
attainable and maximizes $C_1(P_1)$. Correspondingly, $C_1$ spans
the following range:
\begin{align}
C_1(P_{\max},P_{\max}) \leq C_1 \leq C_1(P_{\max},0).
\label{C1rangeb}
\end{align}
So for the range of $\tilde{R}$ in (\ref{RtildeRange}) and the range
of $C_1$ in (\ref{C1rangeb}), $P_1=P_{\max}$ describes the frontier.
Therefore the values of $C_1$ at the frontier are:
\begin{align}
C_1(P_{\max},P_2)=\log_2\left(1+\frac{aP_{\max}}{1+bP_2}\right).
\nonumber
\end{align}
Hence for $C_1(P_{\max},P_{\max}) \leq R \leq C_1(P_{\max},0)$, the
value of $P_2$ that will achieve the frontier follows as:
\begin{align}
P_2=\frac{1}{b}\left(\frac{aP_{\max}}{2^{R}-1}-1\right).
\label{P2range2}
\end{align}
So effectively the value found in (\ref{P2range2}) is the answer for
(\ref{argC2}) for the range of $C_1(P_{\max},P_{\max}) \leq R \leq
C_1(P_{\max},0)$.

\subsection{Achievable Rate Region Frontier}
This subsection consolidates the two results to fully describe the
rate region frontier. For a value of $c_1$ that sweeps the full
range of $C_1$, we have:
\begin{itemize}
\item for $0\leq c_1 \leq C_1(P_{\max},P_{\max})$
\begin{align}
\arg \max_{P_2} C_2(P_2)=P_{\max} \nonumber
\end{align}
and the frontier, denoted by ${\cal F}_2=\Phi(:,P_{\max})$, is
expressed as:
\begin{align}
C_2(c_1)= \log_2\left(1+\frac{cP_{\max}}{\displaystyle
1+\frac{d}{a}(1+bP_{\max})(2^{c_1}-1)}\right) \label{F1}
\end{align}
\item for $C_1(P_{\max},P_{\max})\leq c_1 \leq C_1(P_{\max},0) $
\begin{align}
\arg \max_{P_2}
C_2(P_2)=\frac{1}{b}\left(\frac{aP_{\max}}{2^{c_1}-1}-1\right)
\nonumber
\end{align}
and the frontier, denoted by ${\cal F}_1=\Phi(P_{\max},:)$, is
expressed as:
\begin{align}
C_2(c_1)= \log_2\left(1+\frac{\displaystyle
\frac{c}{b}\left(aP_{\max}-(2^{c_1}-1)\right)}{\displaystyle
(2^{c_1}-1)(1+dP_{\max})}\right).
\end{align}
\end{itemize}
The notation ${\cal F}_i$ denotes a potential line parameterized by
holding the $i^{th}$ element in the power tuple at maximum power.

Finally, the rate region frontier ${\cal F}$ for a two-user
interference channel is obtained as:
\begin{align}
{\cal F}= \mbox{Convex Hull} \{ {\cal F}_1 \cup {\cal F}_2 \}.
\end{align}
The convex hull operation stems from the time sharing solution of
the extremity points in the frontiers in order to arrive to a convex
rate region. For example, in Fig. \ref{fig_twoUserRegion}, ${\cal
F}$ is described by connecting points A and B, and points B and C.

\begin{figure} [t]
\centering
\includegraphics[width=250pt,height=180pt]{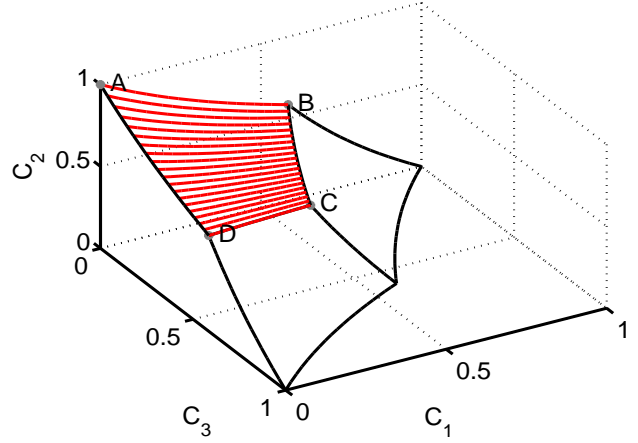}
\caption{A 3-user interference channel rate region}
 \label{fig_3-userRateRegion}
\end{figure}

\section{Achievable Rate Region Frontier for $n-$user Interference
Channel} \label{sec_FrontierGeneral} This section starts by
considering a $3-$user interference channel to show the effect of
adding a new dimension, then generalizes the results for the
$n-$user case.

\subsection{$3-$user example: Effect of increasing $P_3$ from $0$ to $P_{\max}$}
The rate region for the $3-$user case is illustrated in Fig.
\ref{fig_3-userRateRegion}. The following notation of
$\Phi(P_1,P_2,P_3)$ denotes a point in the rate region with
coordinates of $[C_1(P_1,P_2,P_3), C_2(P_1,P_2,P_3),
C_3(P_1,P_2,P_3)]$. Accordingly, $\Phi(:,P_{\max},P_3)$ describes a
line characterized by sweeping the transmit power of the first
transmitter $P_1$ from $0$ to $P_{\max}$, with the second
transmitter transmitting at $P_{\max}$, and the third transmitter
transmitting at a value of $P_3$. Similarly, $C_i(:,P_{\max},:)$
represents a surface in the rate region marked by sweeping the full
range of $P_1$ and $P_3$, and holding $P_2$ at $P_{\max}$.

When $P_3=0$, the same setup and results that were described in
section \ref{sec_FrontierTwoUser} applies. Specifically, for the
rate range of $0\leq C_1 \leq C_1(P_{\max},P_{\max},0)$ and $0\leq
C_2 \leq C_2(0,P_{\max},0)$ and $C_3=0$, the frontier can be
described as $\Phi(:,P_{\max},0)$, which is the line from point A to
point B in Fig. \ref{fig_3-userRateRegion}. As $P_3$ increases, we
want to describe the subsequent effect and how it is traced in the
rate region.

Revisiting the equation in (\ref{Ci}), a constant $P_3$ has the
effect of just an additive noise term in $C_1({\bf P})$ and
$C_2({\bf P})$. Hence, all the previous results in section
\ref{sec_FrontierTwoUser} are applicable for any value of $P_3$ in
describing the frontier for $C_1$ and $C_2$; as the effect of $P_3$
can be lumped in the noise term. Thus for the range of $0\leq C_1
\leq C_1(P_{\max},P_{\max},P_3)$ and $0\leq C_2 \leq
C_2(0,P_{\max},P_3)$, where $P_3$ is constant, the frontier line on
$C_1$ and $C_2$ is $\Phi(:,P_{\max},P_3)$, i.e. characterized by
having $P_2=P_{\max}$. Consequently, the potential lines (or
surfaces) concept in the $3-$user case carries through.

Next, the frontier on $C_3$ needs to be described. For each value of
$P_3$, $\Phi(:,P_{\max},P_3)$ traces one of the highlighted curves
in Fig. \ref{fig_3-userRateRegion}. For these collection of lines to
form a frontier we want to prove that at each increasing value of
$P_3$ these potential lines monotonically increase in the $C_3$ dimension.
This is obvious from the $C_3$ and $P_3$ relation in (\ref{Ci}). The
maximum value of $C_3$ that can be achieved in this case is when
$P_3=P_{\max}$, i.e. $C_3(:,P_{\max},P_{\max})$. Therefore the
highlighted frontier surface in Fig. \ref{fig_3-userRateRegion} is
the potential surface $\Phi(:,P_{\max},:)$. The boundary contours of
this surface are the potential lines: $A\leftrightarrow B$,
$B\leftrightarrow C$, $C\leftrightarrow D$, and $D\leftrightarrow
A$, defined as $\Phi(:,P_{\max},0)$, $ \Phi(P_{\max},P_{\max},:)$,
$\Phi(:,P_{\max},P_{\max})$, and $\Phi(0,P_{\max},:)$, respectively.

By symmetry of interchanging $P_1$, $P_2$ and $P_3$, we find that
the $3-$user rate region frontier is determined through the union of
three surfaces: $\Phi(P_{\max},:,:) \cup \Phi(:,P_{\max},:) \cup
\Phi(:,:,P_{\max})$. And $\cal F$ is expressed as:
\begin{align}
{\cal F} = \mbox{Convex Hull} \{{\cal F}_1 \cup {\cal F}_2 \cup
{\cal F}_3\},
\end{align}
where ${\cal F}_i$ is a potential surface $\Phi(\cdot)$ with
$P_{\max}$ in the $i^{th}$ power position. (Note that the
intersection of potential surfaces is a potential line, as two of
the dimensional inputs become equal, i.e. $\Phi(P_{\max},P_{\max},:) \in
{\cal F}_1$ and $ \Phi(P_{\max},P_{\max},:) \in {\cal F}_2$.)

\subsection{$n-$user generalization}
The case for $n-$user generalization can be done by induction. For
the $n^{th}$ added dimension to the existing $n-1$ dimensions
problem, the additional power effect of $P_n$ can be lumped in the
additive noise term of the existing expressions, and thus the
results for $C_1,\ldots,C_{n-1}$ hold and carry through. The
frontier on $C_n$ is monotonically increasing in $P_n$, and can be
maximized with $P_n=P_{\max}$ for the appropriate range in
$C_1,\ldots,C_{n-1}$. Invoking symmetry we can generalize over all
the rate ranges, therefore arriving to the following theorem.
\begin{theorem}
 The achievable rate region frontier of the $n-$user channel is:
\begin{align}
{\cal F} = \mbox{Convex Hull} \{\cup_{i=1}^n {\cal F}_i\},
\end{align}
where ${\cal F}_i$ is a hyper-surface of $n-1$ dimensions
characterized by holding the $i^{th}$ transmitter at full power
$P_{\max}$.
\end{theorem}
\vspace*{-.2cm} Using the notation introduced in this section,
${\cal F}_i$ is effectively $\Phi(:,\ldots,P_{\max},\ldots,:)$ with
$P_{\max}$ at the $i^{th}$ power position.

\section{Characteristics of the Achievable Rate Region for Two-user Interference Channel}
\label{sec_2usersRegion} Treating the two-user interference channel in more
details,
this section studies the behavior of the rate region frontiers
${\cal F}_1$ and ${\cal F}_2$ in terms of convexity and concavity.
In addition, it discusses when a time-sharing approach should be employed,
with specific results pertaining to the symmetric channel.

\subsection{Convexity or Concavity of the Frontiers}
\begin{figure} [ttt]
\vspace*{-.25cm} \centering
\includegraphics[width=250pt,height=180pt]{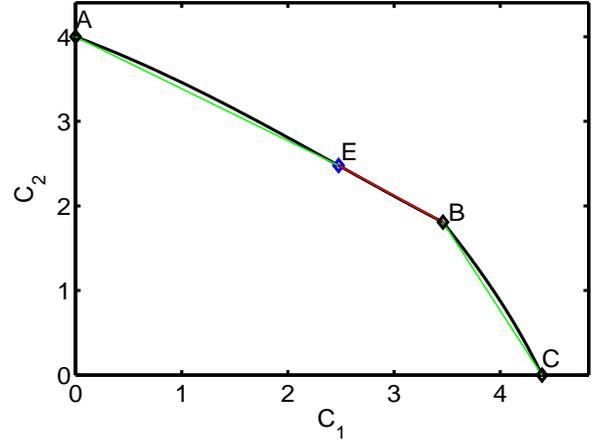}
\vspace*{-.65cm} \caption{Non-stationary inflection point E on
${\cal F}_2$ and concavity on ${\cal F}_1$. ($Q_1=0.45$ and
$Q_2=3.11$, for $P_{\max}=1$, $a=20$, $b=1$, $c=15$, $d=5$.) }
\vspace*{-.55cm} \label{fig_2user_concav_Inflic}
\end{figure}
The frontier ${\cal F}_2$ in (\ref{F1}) depends on $P_1$ through the
following relation of $c_1$ and $P_1$:
\begin{align}
P_1=\frac{1}{a}(1+bP_{\max})(2^{c_1}-1).
\nonumber
\end{align}
Therefore the second derivative of ${\cal F}_2$ with respect to
$c_1$ leads to the following expression:
\begin{align}
\frac{\partial^2 {\cal F}_2}{\partial c_1^2} =
(\theta+adP_1)^2-(a-\theta)(a-\theta+acP_{\max}), \nonumber
\end{align}
where $\theta=d+dbP_{\max}$. Therefore if the frontier line is
concave (i.e. $\frac{\partial^2 {\cal F}_2}{\partial c_1^2} \leq 0$)
then the enclosed region is convex, i.e. the straight line connecting any two
points in the rate region is entirely enclosed in the rate region.
Let $\Re(\cdot)$ be the real operation, and defining
the quantity $Q_1$ as:
\begin{align}
Q_1 = \frac{\Re(\sqrt{(a-\theta)(a-\theta+acP_{\max})})-\theta}{ad},
\end{align}
then it suffices to study the convexity or concavity by examining
the sign of $(P_1-Q_1)$, where $Q_1$ is derived such as:
\begin{align}
\mbox{sign}\left(\frac{\partial^2 {\cal F}_2}{\partial c_1^2}\right)
= \mbox{sign}({P_1-Q_1}).\nonumber
\end{align}
Thus the convexity or concavity of the frontier line ${\cal F}_2$ is
governed by:
\begin{itemize}
\item $Q_1 \leq 0$: the frontier line ${\cal F}_2$ is convex, and
the region bounded by ${\cal F}_2$ is concave. As $P_1-Q_1 \geq 0$
for all the range of $P_1$.
\item $Q_1 \geq P_{\max}$: the frontier line ${\cal F}_2$ is concave, and
 the region bounded by ${\cal F}_2$ is convex. As $P_1-Q_1 \leq 0$ for
 all the range of $P_1$.
\item $0 < Q_1 < P_{\max}$: the frontier line exhibits a non-stationary
inflection point when $P_1=Q_1$, and ${\cal F}_2$ is neither convex
nor concave between the point extremities of $\Phi(0,P_{\max})$ and
$\Phi(P_{\max},P_{\max})$. In this case:
    \begin{itemize}
     \item for $0 < P_1 \leq Q_1$: the
     line $\Phi(0:Q_1,P_{\max})$ is concave, i.e. the frontier segment between
     point $\Phi(0,P_{\max})$ (the point A in Fig. \ref{fig_2user_concav_Inflic})
     and point $\Phi(Q_1,P_{\max})$ (the point E in Fig.
     \ref{fig_2user_concav_Inflic}) is concave.
     \item for $Q_1 \leq P_1 < P_{\max}$: the line $\Phi(Q_1:P_{\max},P_{\max})$
     is convex, i.e. the frontier segment between point $\Phi(Q_1,P_{\max})$
     (the point E in Fig. \ref{fig_2user_concav_Inflic})
     and point $\Phi(P_{\max},P_{\max})$ (the point B in Fig.
     \ref{fig_2user_concav_Inflic}) is convex.
\end{itemize}
\end{itemize}
By symmetry, the frontier line ${\cal F}_1$ exhibits the following
behavior: it is convex when $Q_2 \leq 0$, and it is concave when
$Q_2 \geq P_{\max}$, and it exhibits an non-stationary inflection
point when $P_2=Q_2$ -- specifically it is convex for $0 < P_2 \leq
Q_2$ and concave for $Q_2 \leq P_2 < P_{\max}$. Hereby $Q_2$ is
defined as:
\begin{align}
Q_2 = \frac{\Re(\sqrt{(c-\beta)(c-\beta+acP_{\max})})-\beta}{cb},
\end{align}
with $\beta=(b+bdP_{\max})$.

When describing the full rate region frontier through ${\cal F}_1
\cup {\cal F}_2$, the rate region is convex if both ${\cal F}_1$ and
${\cal F}_2$ are concave, and the rate region is concave otherwise.
Therefore, whenever the frontier (or segment thereof) is concave, it
will describe the convex hull of the rate region instead of a time
sharing solution. Fig. \ref{fig_2user_concav_Inflic} illustrates an
example where the frontier
 ${\cal F}_1$ is concave, and the frontier ${\cal F}_2$ exhibiting a
 non-stationary inflection point E.
In this case the convex hull rate region is found by operating along
the concave frontier ${\cal F}_1$, and time-sharing between point B
and point E, and operating along the concave segment of ${\cal F}_2$
between point E and point A.

\subsection{Optimality of Time Sharing}
This subsection investigates the optimality of time sharing between
operating points in the rate region. For instance, whenever the rate
region frontier segment is convex (equivalently, the enclosed rate
region is concave) then operating with time sharing between the
extremities of the curve is optimal than operating along the
log-defined frontier. Analyzing the ${\cal F}_2$ frontier, and
referring to the Fig. \ref{fig_twoUserRegion}, it follows:
\begin{itemize}
\item $Q_1 \leq 0$: (i.e. ${\cal F}_2$ frontier is convex)
it is optimal to apply time-sharing through the following options:
    \begin{itemize}
     \item between point A and point B,
    \item between point A and point $\Phi(P_{\max},Q_2)$ if ${\cal F}_1$
          exhibits a non-stationary inflection point,
    \item between point A and a point on the concave segment of
          ${\cal F}_1$,
    \item between point A and point C.
\end{itemize}
These depend on how the parameters $a,b,c,d,$ and $P_{\max}$ would
lead to a convex hull region. This can be done by evaluating and
comparing each of the candidate solution aforementioned.

\item $Q_1 \geq P_{\max}$: (i.e. ${\cal F}_2$ frontier is concave)
the potential line $\Phi(:,P_{\max})$ is
optimal, and no time sharing is to be employed.
\item $0 < Q_1 < P_{\max}$: it is optimal to use the concave potential
 line segment $\Phi(0:Q_1,P_{\max})$, and subsequently to use the time
 sharing candidate options that the case $Q_1 \leq 0$ mentioned but with
  replacing the point A with the point $\Phi(Q_1,P_{\max})$.
\end{itemize}
The conditions for ${\cal F}_1$ follow by symmetry.

\subsubsection{Time sharing between Points A, B, and C in Fig.
\ref{fig_twoUserRegion}} Discounting the case when ${\cal F}_1$ or
${\cal F}_2$ exhibit non-stationary inflection point for simplicity,
and focusing on the case of $Q_1 \leq 0$ and $Q_2 \leq 0$, it is
important to know when time sharing between point A and point C is
better than time sharing through the intermediate point B. This is
done by comparing the straight line connecting points A and C, and
the coordinates of B. It follows that operating with time sharing
between the points (or system states) A and C (i.e. one transmitter
only transmitting at a certain point) is optimal when:
\begin{align}
\frac{(1+cP_{\max})(1+dP_{\max})}{1+cP_{\max}+dP_{\max}} \geq
\left(\frac{1+aP_{\max}+bP_{\max}}{1+bP_{\max}}\right)^\gamma
\label{condAC}
\end{align}
with $\gamma=\log_2(1+cP_{\max})/\log_2(1+aP_{\max})$.

\subsubsection{Symmetric two-user interference channel}
\begin{figure} [ttt]
\centering
\includegraphics[width=250pt,height=180pt]{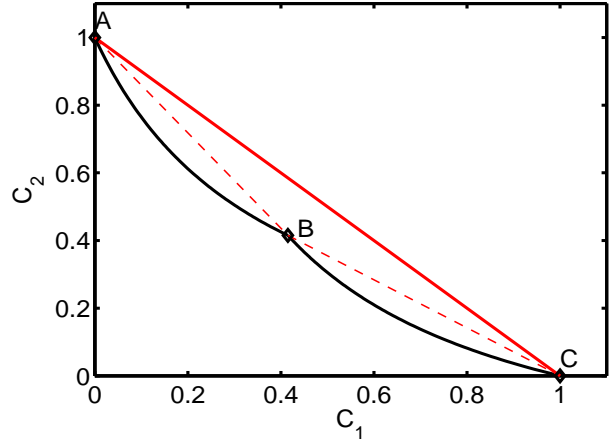}
\caption{Optimal time sharing between point A and point C. (For $P_{\max}=1$,
$a=1$, $b=2$, and $b^*$ calculated from (\ref{bsym}) equals $\sqrt{2}$.)}
 \label{fig_twoUserSym}
 \vspace*{-.5cm}
\end{figure}
For the symmetric two-user interference channel, $a=c$, and $b=d$.
In this case, (\ref{condAC}) simplifies and leads to the following
theorem:
\begin{theorem} Time-sharing operation with one transmitter active
 at full power at a time is
optimal when
\begin{align}
b \geq \frac{\displaystyle \sqrt{1+ aP_{\max}}}{P_{\max}}.
\label{bsym}
\end{align} \label{theorem_b}
\end{theorem}
\emph{Remark:} For high SNR (i.e. $aP_{\max} \gg 1$), (\ref{bsym})
reduces to $bP_{\max} \geq \sqrt{aP_{\max}}$, which interestingly
coincides with recent results in
\cite[eq.(3)]{EtkinTse:GaussianInterfChannelToOneBit}.

This indicates that when the cross interfering power gain $b$
exceeds the right hand side value in (\ref{bsym}), denoted by
$b^{\ast}$, it is optimal to operate with one transmitter at a time.
Fig. \ref{fig_twoUserSym} uses $b=2$, which is larger than
 $b^*=\sqrt{2}$ obtained from (\ref{bsym}) for
$a=1$ and $P_{\max}=1$. By contrast in Fig. \ref{fig_twoUserRegion}, the
value of $b=1$ is adopted and the behavior exhibited is different as
expected.

In addition, we subsequently prove that the expression in
(\ref{bsym}) is a sufficient condition for both frontiers ${\cal
F}_1$ and ${\cal F}_2$ to be convex, i.e. $Q_1$ and $Q_2$ are always
$\leq 0$.

\begin{proof} For the symmetric case, $Q_1=Q_2=Q_{sym}$ can be written as
\begin{align}
Q_{sym}
=\frac{\Re(\sqrt{(a-\theta)(a-\theta+a^2P_{\max})})-\theta}{ab},
\nonumber
\end{align}
where $\theta=b+b^2P_{\max}$. $Q_{sym}$ can also be written in this
form:
\begin{align}
Q_{sym} =\frac{\Re(\sqrt{T_1T_2})-\theta}{ab}, \nonumber
\end{align}
where $T_1=a-\theta=a-b-b^2P_{\max}$, and
$T_2=a-\theta+a^2P_{\max}$. From the expression in (\ref{bsym}), $a$
can be alternatively upper-bounded as $a\leq
(b^2P_{\max}-1/P_{\max})$. Therefore, $T_1$ is upper-bounded as:
\begin{align}
T_1 \leq -1/P_{\max} - b \label{T1}.
\end{align}
From (\ref{T1}), $T_1$ is always negative. $T_2$ however can be
positive or negative. Each case is evaluated as follows:
\begin{itemize}
\item $T_2\geq 0$: then $\Re(\sqrt{T_1T_2})=0$, and
as $\theta$ is always positive, then $Q_{sym}\leq 0$.
\item $T_2\leq 0$: $\Re(\sqrt{T_1T_2})$ is $ \geq 0$. In this case,
the numerator of $Q_{sym}$ can be written as:
\begin{align}
\mbox{num}(Q_{sym}) = \sqrt{(\theta-a)(\theta-a-a^2P_{\max})}-\theta.
\nonumber
\end{align}
Given the fact that $(\theta-a-a^2P_{\max}) \leq (\theta-a)$, then
$\mbox{num}(Q_{sym})$ can be upper-bounded as:
\begin{align}
\begin{array}{ll}
\mbox{num}(Q_{sym}) & \leq \sqrt{(\theta-a)^2}-\theta \\
& \leq -a \leq 0.
\end{array}\nonumber
\end{align}
\end{itemize}
Hence the frontiers ${\cal F}_1$ and ${\cal F}_2$ are convex.
\end{proof}
%
Therefore, when $b$ satisfies the equation in (\ref{bsym}) the
frontiers lines are always convex, and time-sharing with only one
transmitter active at a time is optimal.

\section{Conclusions}
The achievable rate region frontiers for the $n-$user interference channel
were presented when there is no cooperation at the transmit side nor at the
receive side. The receivers do not employ multiuser detection, and
the interference is considered as additive noise. Results were first
found for the two-user interference channel. The $3-$user interference channel
was treated next to show the effect of adding the additional dimension, and
subsequently the $n-$user interference channel generalization results
followed. The $n-$user rate region is found to be the convex hull of the
union of $n$ hyper-surfaces each of dimension $n-1$. Each hyper-surface
frontier ${\cal F}_i$ is defined by having the $i^{th}$ transmitter
transmitting at its full power $P_{\max}$.

The two-user interference channel was further studied regarding the
convexity or the concavity of the frontiers. Conditions when the
frontier is convex or concave or exhibiting a non-stationary
inflection point were also obtained. Whenever the log-defined
frontier line is convex then a time sharing solution is optimal. For
the symmetric two-user case, the condition was found to indicate
when time-sharing between the points $\Phi(P_{\max},0)$ and
$\Phi(0,P_{\max})$ (i.e. one transmitter solely transmitting at full
power at a certain time) is to be used, rather than time-sharing
through the point $\Phi(P_{\max},P_{\max})$ (i.e. both transmitters
transmitting at full power). That condition was also proven to be
sufficient to ensure that both frontiers (${\cal F}_1$ and ${\cal
F}_2$) will in fact always be convex.

\appendix
Proof that the equation (\ref{C2P2}), $C_2(P_2)$ is monotonically
increasing in $P_2$.\\
\begin{proof}
Effectively (\ref{C2P2}) is in the form of $f(1+g(x))$. As $f(\cdot)$ is monotonically
increasing in its argument, it suffices to prove that $g(x)$ is
monotonically increasing in $x$. Therefore define $g(P_2)$ as:
\begin{equation}
g(P_2) = \frac{acP_2}{a+d(1+bP_2)(2^R-1)}, \nonumber
\end{equation}
\begin{eqnarray}
\lefteqn{\frac{\partial g(P_2)}{\partial P_2} =}\nonumber\\
& & \frac{ac}{a+d(1+bP_2)(2^R-1)} -
\frac{acP_2db(2^R-1)}{(a+d(1+bP_2)(2^R-1))^2} \nonumber\\
& & = \frac{a^2c+acd(1+bP_2)(2^R-1)-acP_2db(2^R-1)}{(a+d(1+bP_2)(2^R-1))^2} \nonumber\\
& & = \frac{a^2c+acd(2^R-1)+acdbP_2(2^R-1)-acdbP_2(2^R-1)}{(a+d(1+bP_2)(2^R-1))^2}\nonumber\\
& & = \frac{ac(a+d(2^R-1))}{(a+d(1+bP_2)(2^R-1))^2}. \label{del_g}
\end{eqnarray}
The numerator in (\ref{del_g}) is $\neq 0$ if $a \neq 0$ and $c \neq 0$ ($a=0$ or $c=0$
are the trivial cases where the rate region is either a line or the point zero).
As $R\geq 0$, then $(2^R-1) \geq0$. Thus $\partial g(P_2)/{\partial P_2}$ is
always $> 0$ for non-trivial cases of $a$ and $c$. Thus $g(P_2)$ is monotonically increasing in $P_2$,
and equivalently $C_2(P_2)$ is monotonically increasing in $P_2$.
\end{proof}


\bibliography{bib_entryM}

\end{document}